\documentclass[11pt,letter]{article}

\usepackage{fullpage}

\usepackage{hyperref}
\usepackage{verbatim}
\usepackage{amsmath}
\usepackage{latexsym}
\usepackage{revsymb}
\usepackage{yfonts}
\usepackage{amsfonts}
\usepackage{amsmath}
\usepackage{amssymb}
\usepackage{amsthm,amssymb,amsfonts}
\usepackage{graphicx}
\usepackage{color}
\usepackage{mathrsfs}

\newenvironment{sciabstract}{%
\begin{quote} \bf}
{\end{quote}}
\newcounter{lastnote}

\title{ {The uncertainty principle determines the non-locality of quantum mechanics}}
\author
{Jonathan Oppenheim$^{1}$, Stephanie Wehner$^{2,3}$\\
\\
\normalsize{$^{1}$DAMTP, University of Cambridge, CB3 0WA, Cambridge, UK}\\ 
\normalsize{$^{2}$Institute for Quantum Information, Caltech, Pasadena, CA 91125, USA}\\
\normalsize{$^{3}$Centre for Quantum Technologies, National University of Singapore, 117543 Singapore}\\
\\
}
\date{}
\newcommand{\be}{\begin{eqnarray} \begin{aligned}}
\newcommand{\ee}{\end{aligned} \end{eqnarray} }
\newcommand{\benn}{\begin{eqnarray*} \begin{aligned}}
\newcommand{\eenn}{\end{aligned} \end{eqnarray*} }

\newcommand{\bc}{\begin{center}}
\newcommand{\ec}{\end{center}}

\newcommand{\id}{\mathbb{I}}

\newcommand{\tr}{\mathop{\mathsf{tr}}\nolimits}

\newtheorem{theorem}{Theorem}[section]

\newtheorem{lemma}[theorem]{Lemma}
\newtheorem{claim}[theorem]{Claim}

\newtheorem{corollary}[theorem]{Corollary}

	\newcommand{\myacknowledgments}{{\bf Acknowledgments\ }}

\usepackage{amsfonts}
\def\Real{\mathbb{R}}

\def\id{\mathbb{I}}

\def\01{\{0,1\}}

\newcommand{\floor}[1]{\lfloor{#1}\rfloor}
\newcommand{\eps}{\varepsilon}
\newcommand{\ket}[1]{|#1\rangle}
\newcommand{\bra}[1]{\langle#1|}

\newcommand{\mB}{\mathcal{B}}

\newcommand{\ens}{\mathcal{E}}

\newcommand{\psucRAC}{P^{\rm cert}}
\newcommand{\psuc}{P^{\rm succ}}
\newcommand{\pgame}{{P^{\rm game}(\setS,\setT,\sigma_{AB})}}
\newcommand{\pgameMAX}{{P^{\rm game}_{\rm max}}}

\newcommand{\cU}{\mathcal{U}}

\newcommand{\hmin}{\ensuremath{H}_{\infty}}

\newcommand{\zetastr}{\zeta_\str}
\newcommand{\zetamax}{\zeta_\str}
\newcommand{\zetamaxx}{\zeta_\xstr}
\newcommand{\zetaset}{\zeta_\str^\Sigma}
\newcommand{\zetasteer}{\zeta_\str^{\sigma_{s,a}}}
\newcommand{\stateSet}{\mathscr{S}}

\newcommand{\xstr}{{\vec{x}_{s,a}}}
\newcommand{\str}{{\vec{x}}}
\newcommand{\vstr}{x}
\newcommand{\vxstr}{{x_{s,a}}}

\newcommand{\prob}[1]{p(#1)}

\newcommand{\xor}{XOR}

\newcommand{\unop}{Q}

\newcommand{\qBob} {{\mathcal{D}}}
\newcommand{\setA}{\mathcal{A}}
\newcommand{\setB}{\mathcal{B}}
\newcommand{\setS}{\mathcal{S}}
\newcommand{\setT}{\mathcal{T}}
\newcommand{\setTopt}{\mathcal{T}_{\rm opt}}

\newcommand{\post}{\tau_{t,b}}
\newcommand{\ppost}[1]{{p(#1)_{\post}}}

\begin{document} 
\maketitle 
\begin{sciabstract}
Two central concepts of quantum mechanics are Heisenberg's uncertainty
principle, and 
a subtle form of non-locality that Einstein famously
called ``spooky action at a distance''.  
These two fundamental features have thus far been distinct concepts.
Here we show that they are inextricably and quantitatively linked.  
Quantum mechanics cannot be 
more non-local with measurements that respect the uncertainty principle.  
In fact,
the link between uncertainty and non-locality holds for all physical theories.
More specifically, the degree of non-locality of any theory is determined
by two factors -- 
the strength of the uncertainty principle, 
and the strength of a property
called ``steering'', which determines which
states can be prepared at one location given a measurement at another.

\end{sciabstract}

A measurement allows us to gain information about the state of a physical system.
For example when measuring the position of a particle, the measurement outcomes correspond to possible locations.
Whenever the state of the system is such that we can predict this position with certainty, then there is 
only one measurement outcome that can occur.
Heisenberg~\cite{heisenberg:uncertainty} observed that quantum mechanics
imposes strict restrictions on what we can hope to learn -- there are incompatible measurements
such as position and momentum whose results cannot be simultaneously predicted 
with certainty. These restrictions are known as uncertainty relations.
For example, uncertainty relations tell us that if we were able to predict
the momentum of a particle with certainty, then when measuring its position all measurement
outcomes occur with equal probability. That is, we are completely uncertain about its location.

Non-locality can be exhibited when performing measurements on two or more distant quantum systems -- the outcomes
can be correlated in way that defies any local classical description~\cite{bell:epr}.  This is why we know that quantum theory will never by superceded by a local classical theory.
Nevertheless, even quantum correlations are restricted
to some extent -- measurement results cannot be correlated
so strongly that they would 
allow signalling between two distant systems.  However
quantum correlations are still weaker
than what the no-signalling principle demands~\cite{PR,PR1,PR2}.  So, why are
quantum correlations strong enough to be non-local, yet not as strong
as they could be?  Is there a principle that determines the degree of this non-locality?
Information-theory~\cite{gs:relaxedUR,infoCausality},
communication complexity~\cite{wim:nonlocal}, 
and local quantum mechanics~\cite{s3:nonlocal}
provided us with some rationale why 
limits on quantum theory may exist.
But evidence suggests that many of these attempts
provide only partial answers.
Here, we take a very different approach and relate the limitations of non-local correlations to two 
inherent properties of any physical theory.

\section*{Uncertainty relations}
At the heart of quantum mechanics lies 
Heisenberg's uncertainty principle~\cite{heisenberg:uncertainty}.
Traditionally, uncertainty relations have been expressed
in terms of commutators
\begin{align}
	\Delta A \Delta B \geq  \ \frac{1}{2} |\bra{\psi}[A,B]\ket{\psi}|.
\end{align}
with 
standard deviations $\Delta X = \sqrt{\bra{\psi}X^2\ket{\psi} - \bra{\psi}X\ket{\psi}^2}$ for $X \in \{A,B\}$.
However, the more modern approach is to use entropic
measures.
Let $p(x^{(t)}|t)_\sigma$ denote the probability 
that we obtain outcome $x^{(t)}$ when performing a measurement labelled $t$ when the system
is prepared in the state $\sigma$.
In quantum theory, $\sigma$ is a density operator, while for a general physical theory, we assume that $\sigma$
is simply an abstract representation of a state.
The well-known Shannon entropy $H(t)_\sigma$ of the distribution over measurement outcomes of measurement $t$ on a system in state $\sigma$ is thus
\begin{align}
H(t)_\sigma := - \sum_{x^{(t)}} p(x^{(t)}|t)_\sigma \log p(x^{(t)}|t)_\sigma\ .
\end{align}
In any uncertainty relation, we wish to compare outcome distributions for multiple measurements. In terms of
entropies such relations are of the form
\begin{align}
\sum_{t \in \setT} 
p(t)\ H(t)_\sigma
&\geq c_{\setT,\qBob} ,
\label{eq:entropicUR}
\end{align}
where 
$p(t)$ is any probability distribution over the set of measurements $\setT$,
and $c_{\setT,\qBob}$ is some positive constant determined by $\setT$ and the distribution $\qBob = \{p(t)\}_t$. 
To see why~\eqref{eq:entropicUR} forms an uncertainty relation, note that whenever $c_{\setT,\qBob}> 0$
we cannot predict the outcome of at least one of the measurements $t$ with certainty, i.e., $H(t)_\sigma > 0$.
Such relations 
have the great advantage 
that the lower bound $c_{\setT,\qBob}$ does not depend on the state $\sigma$~\cite{deutsch:uncertainty}. Instead,
$c_{\setT,\qBob}$ depends only on the measurements and hence quantifies their
inherent incompatibility. It has been shown that for two measurements,
entropic uncertainty relations do in fact imply Heisenberg's uncertainty relation~\cite{bm:uncertainty}, providing
us with a more general way of capturing uncertainty (see~\cite{ww:URsurvey} for a survey).

One may consider many entropic measures instead of the Shannon entropy. For example, the min-entropy
\begin{align}
\hmin(t)_\sigma := - \log \max_{x^{(t)}} p(x^{(t)}|t)_\sigma\ 
\end{align}
used in~\cite{deutsch:uncertainty}, plays an important role in cryptography 
and provides a lower bound on $H(t)_\sigma$.
Entropic functions are, however, a rather coarse way of measuring the uncertainty of a set of measurements,
as they
do not distinguish the uncertainty inherent in obtaining any combination of outcomes $\vstr^{(t)}$ for different measurements $t$.
It is thus
useful to consider much more fine-grained uncertainty relations consisting of a series of inequalities, one
for each combination of possible outcomes, which we write as a string
$\str = (\vstr^{(1)},\ldots,\vstr^{(n)}) \in \setB^{\times n}$ with $n = |\setT|$~\footnote{Without loss of generality we assume that each measurement has the same set of possible outcomes, since we may simply add additional outcomes which never occur.}.
That is, for each $\str$, a set of measurements $\setT$, and distribution $\qBob = \{p(t)\}_t$, 
\begin{align}\label{eq:uncertRelOne}
	\psucRAC(\sigma;\str) := \sum_{t=1}^n p(t)\ p(x^{(t)}|t)_\sigma \leq \zetastr(\setT,\qBob)\ .
\end{align}

For a fixed set of measurements, the set of inequalities
\begin{align}\label{eq:urRelations}
\cU = \left\{\sum_{t=1}^{n}  p(t)\ p(x^{(t)}|t)_\sigma\ \leq \zetastr \mid \forall \str \in \setB^{\times n} \right\}\ ,
\end{align}
thus forms a fine-grained uncertainty relation, as it dictates that 
one cannot obtain a measurement outcome with certainty
for all measurements simultaneously whenever
$\zetastr < 1$.  
The values of $\zetastr$
thus confine the set of 
allowed probability distributions,
and the measurements have uncertainty if 
$\zetastr < 1$ for all $\str$. 
To characterise the ``amount of uncertainty'' in a particular physical theory, we are thus interested in the values of
\begin{align}
\zetamax = \max_{\sigma} \sum_{t=1}^n p(t) p(x^{(t)}|t)_\sigma\
\label{eq:urRelationsmax}
\end{align}
where the maximization is taken over all states allowed on a particular 
system (for simplicity, we assume it can be attained in the theory considered).
We will also refer to the state $\rho_{\str}$ that attains the maximum as a ``maximally certain state''.  However, we will also be interested in the degree of uncertainty exhibited by
measurements on a set of states $\Sigma$ quantified by
$\zetaset$ defined with the maximisation in Eq. \ref{eq:urRelationsmax} taken over ${\sigma \in \Sigma}$.
Fine-grained uncertainty relations are directly related
to the entropic ones, and have both a physical and an information processing appeal (see appendix).   As an example, consider 
the binary spin-$1/2$ observables $Z$ and $X$.  If 
we can obtain a particular outcome $x^{(Z)}$ given that we
measured $Z$ with certainty i.e. $p(x^{(Z)}|Z)=1$, then the complementary
observable must be completely uncertain i.e.  $p(x^{(X)}|X)=1/2$.
If we choose which
measurement to make with probability $1/2$ then      
this notion of uncertainty is captured by the relations
\begin{align}
\frac{1}{2}p(x^{(X)}|X)+\frac{1}{2}p(x^{(Z)}|Z)\leq
\zetamax
= \frac{1}{2} + \frac{1}{2\sqrt{2}} \mbox{ for all } \str=(x^{(X)},x^{(Z)})  \in \01^2\ ,
\label{eq:uncertaintyexample}
\end{align}
where the maximally certain states are given by the eigenstates of $(X+Z)/\sqrt{2}$ and
$(X - Z)/\sqrt{2}$.

\section*{Non-local correlations}
We now introduce the concept of non-local correlations.
Instead of considering measurements on a single system we consider
measurements on two (or more)
space-like separated systems traditionally
named Alice and Bob.
We label Bob's measurements using $t \in \setT$, and use $b \in \setB$ to label his measurement outcomes.
For Alice, we use $s \in \setS$ to label her measurements, and $a \in \setA$ to label her outcomes (recall, that wlog we can assume all measurements have the same number of outcomes).
When Alice and Bob perform measurements on a shared state $\sigma_{AB}$
the outcomes of their measurements can be correlated. 
Let $p(a,b|s,t)_{\sigma_{AB}}$ be the joint probability that they obtain outcomes $a$ and $b$ for measurements
$s$ and $t$. We can now again ask ourselves, 
what correlations are possible in nature? In other words,
what probability distributions are allowed?

Quantum mechanics as well as classical mechanics obeys the no-signalling principle, meaning that information cannot travel faster
than light. This means that the probability that Bob obtains outcome $b$ when performing measurement $t$ cannot depend on which measurement
Alice chooses to perform (and vice versa).
More formally, $\sum_{a} p(a,b|s,t)_{\sigma_{AB}}= p(b|t)_{\sigma_B}$ for all $s$, where $\sigma_B = \tr_A(\sigma_{AB})$.
Curiously, however, this is not the only constraint imposed by quantum mechanics~\cite{PR},
and finding other constraints is a difficult task~\cite{acin:bell,as:SDPhierarchy}.
Here, we find that the uncertainty principle imposes the other limitation.

To do so, let us recall the concept of so-called Bell inequalities~\cite{bell:epr},
which are
used to describe limits on such joint probability distributions. 
These are most easily explained in their more modern form in terms of a game played between Alice and Bob.
Let us choose questions $s \in \setS$ and $t \in \setT$ according to some distribution $p(s,t)$ and send them to Alice
and Bob respectively, where we take for simplicity $p(s,t) = p(s) p(t)$.
The two players now return answers $a \in \setA$ and $b \in \setB$. 
Every game comes with a set of rules that determines whether $a$ and $b$ are winning answers given questions
$s$ and $t$. Again for simplicity, we thereby assume 
that for every $s$,$t$ and $a$, there exists exactly one winning
answer $b$ for Bob (and similarly for Alice). 
That is, for every setting $s$ and outcome $a$ of Alice
there is a string 
$\xstr = (\vxstr^{(1)},\ldots,\vxstr^{(n)}) \in \setB^{\times n}$ of length $n = |\setT|$ 
that determines the correct answer $b = \vxstr^{(t)}$ for question $t$ for Bob (Fig.~\ref{fig:gameRules}). We say that $s$ and $a$ determine a ``random access coding''~\cite{as:dimBound}, meaning that Bob is not trying to learn the full
string $\xstr$ but only the value of one entry.
The case of non-unique games is a straightforward but cumbersome generalisation.

As an example, the Clauser-Horne-Shimony-Holt (CHSH) inequality~\cite{chsh:inequality}, one of the simplest
Bell inequalities whose violation implies non-locality, can be expressed as a game
in which Alice and Bob receive
binary questions $s, t \in \01$ respectively, and similarly their answers
$a, b\in \01$ are single bits.  Alice and Bob win the CHSH game if their answers satisfy
$a \oplus b = s \cdot t$. 
We can label Alice's outcomes using string $\xstr$ and Bob's goal is to retrieve the $t$-th element of this string. 
For $s=0$, Bob will always need to give the same answer as Alice in order to win 
and hence we have
$\vec{x}_{0,0} = (0,0)$, and $\vec{x}_{0,1} = (1,1)$. For $s=1$, Bob needs to give the same answer for $t=0$, but the opposite
answer if $t=1$. That is, $\vec{x}_{1,0} = (0,1)$, and $\vec{x}_{1,1} = (1,0)$.

\begin{figure}
\begin{center}
\includegraphics[scale=1.5]{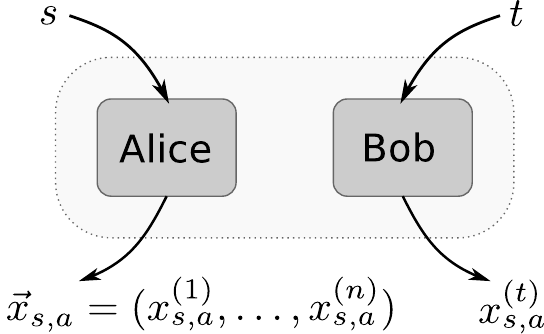}
\caption{Any game where for every choice of settings $s$ and $t$, and answer $a$ of Alice, there is only one 
correct answer $b$ for Bob can be expressed using strings $\xstr$ such that $b = \vxstr^{(t)}$ is 
the correct answer for Bob~\cite{as:dimBound} for $s$ and $a$. The game is equivalent to
Alice outputting $\xstr$ and Bob outputting $\vxstr^{(t)}$.}
\label{fig:gameRules}
\end{center}
\end{figure}

Alice and Bob may agree on any strategy ahead of time, but once the game
starts their physical distance prevents them from communicating.
For any physical theory, such a strategy consists of a choice of shared state $\sigma_{AB}$, as
well as measurements, where we may without loss of generality assume that the players
perform a measurement depending on the question they receive and return the outcome of said measurement
as their answer. 
For any particular strategy, we may hence write the probability that Alice and Bob 
win the game as
\begin{align}
\pgame
= \sum_{s,t} p(s,t) \sum_{a} 
p(a,b = \vxstr^{(t)}|s,t)_{\sigma_{AB}}
\label{eq:pgame}
\end{align}

To characterize what distributions are allowed, we are generally interested in the winning probability maximized
over all possible strategies for Alice and Bob
\begin{align}
\pgameMAX =
\max_{\setS,\setT,\sigma_{AB}} \pgame\ ,
\end{align}
which in the case of quantum theory, we refer to as a Tsirelson’s type bound
for the game~\cite{tsirel:original}. For the CHSH inequality, we have 
$\pgameMAX = 3/4$ classically, $\pgameMAX = 1/2 +1(2\sqrt{2})$ quantum mechanically,
and $\pgameMAX= 1$ for a theory allowing any nonsignalling correlations. 
$\pgameMAX$ quantifies the strength of nonlocality for any theory, with the
understanding that a theory possesses genuine nonlocality
when it differs from the value that can be achieved classically.
The connection we will demonstrate between uncertainty relations
and nonlocality holds even before this optimization.

\section*{Steering}
The final concept, we need in our discussion is steerability, which determines
what states Alice can prepare on Bob's system remotely. Imagine Alice and Bob share a state $\sigma_{AB}$, 
and consider the reduced state $\sigma_B = \tr_A(\sigma_{AB})$ on Bob's side.
In quantum mechanics, as well as many other theories in which Bob's state space is a convex set $\stateSet$,
the state $\sigma_B \in \stateSet$ can be decomposed in many different ways as a convex sum
\begin{align}
\sigma_B=\sum_a p(a|s)\ \sigma_{s,a} \mbox{ with } \sigma_{s,a} \in \stateSet\ ,
\label{eq:nosig}
\end{align}
corresponding to an ensemble $\mathcal{E}_s = \{p(a|s),\sigma_{s,a}\}_a$.
It was Schr\"{o}dinger~\cite{schroedinger:steering,schroedinger:steering2} who noted that in quantum mechanics for all $s$
there exists a measurement on Alice's system that allows her to prepare $\mathcal{E}_s = \{p(a|s), \sigma_{s,a}\}_a$ on Bob's site (Fig.~\ref{fig:steering}).
That is for measurement $s$, Bob's system will be in the state $\sigma_{s,a}$ with probability $p(a|s)$.  Schr\"{o}dinger called this steering to the ensemble $\mathcal{E}_s$ and it does not violate the no-signalling principle,
because for each of Alice's setting, the state of Bob's system
is the same once we average over Alice's measurement outcomes.

Even more, he observed that for any set of ensembles
$\{\mathcal{E}_s\}_s$ 
that respect the no-signalling constraint, i.e.,
for which there exists a state $\sigma_B$ such that Eq. \ref{eq:nosig} holds,
we can in fact find a bipartite quantum state $\sigma_{AB}$ and measurements that allow Alice to steer to such ensembles.
\begin{figure}
\begin{center}
\includegraphics[scale=1.5]{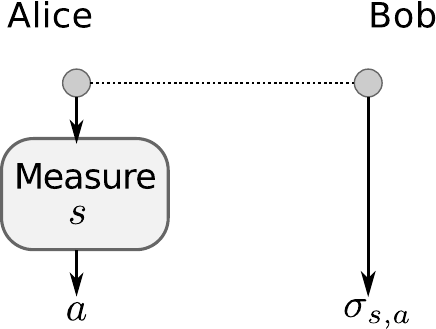}
\caption{When Alice performs a measurement labelled $s$, and obtains outcome $a$ with probability $p(a|s)$ she effectively 
prepares the state $\sigma_{s,a}$ on Bob's system.  This is known as ``steering''.}
\label{fig:steering}
\end{center}
\end{figure}
We can imagine theories in which our ability to steer is either more, or maybe even
less restricted (some amount of signalling is permitted).
Our notion of steering thus allows forms of steering not considered
in quantum mechanics~\cite{schroedinger:steering,schroedinger:steering2,andrew:steering1,andrew:steering2} 
or other restricted classes of theories~\cite{barnum2009ensemble}.
Our ability to steer, however, is a property of the set of ensembles we consider, and not a property
of one single ensemble.

\section*{Main result}

We are now in a position to derive the relation between how 
non-local a theory is, and how uncertain it is.
For any theory, we will find that the uncertainty relation for Bob's measurements (optimal or otherwise) acting on the states that Alice can steer to is what determines the strength of non-locality (See Eq. (\ref{eq:optimalpgame})).  We then find that in quantum mechanics,
for all games where the optimal measurements are known, the states which Alice can steer to are identical to the most certain states, and it is thus only the uncertainty relations of Bob's measurements which determine the outcome (see Eq. (\ref{eq:XORexact})).

First of all, note that we may express the probability~\eqref{eq:pgame} that Alice and Bob win the game using a particular strategy (Fig~\ref{fig:gameRules}) as
\begin{align}\label{eq:gameRewrite}
\pgame
= \sum_s p(s) \sum_a p(a|s) \psucRAC(\sigma_{s,a};\xstr)\ ,
\end{align}
where $\sigma_{s,a}$ is the reduced state of Bob's system for setting $s$ and outcome $a$ of Alice,
and $p(t)$ in $\psucRAC(\cdot)$ is the probability distribution over Bob's questions $\mathcal{T}$ in 
the game. 
This immediately suggests that there is indeed a close connection between games and our fine-grained uncertainty relations.
In particular, every game can be understood as giving rise to a set of uncertainty relations and vice versa.
It is also 
clear from Eq.~\ref{eq:uncertRelOne} for Bob's choice of measurements $\setT$ and distribution $\qBob$ over $\setT$
that the strength of the uncertainty relations imposes an upper bound on the winning probability
\begin{align}\label{eq:gameRewriteURBound}
\pgame
&\leq \sum_s p(s) \sum_a p(a|s)\ \zeta_{\xstr}(\setT,\qBob)
 \leq \max_{\xstr} \zeta_{\xstr}(\setT,\qBob)\ ,
\end{align}
where we have made explicit the functional dependence of $\zetastr$ on the set of measurements.
This seems curious given that we know from~\cite{wolf:ur} that any two incompatible
observables lead to a violation of the CHSH inequality, and that to achieve Tsirelson's bound
Bob has to measure maximally incompatible observables~\cite{peres:book} that yield stringent
uncertainty relations.
However, from Eq.~\ref{eq:gameRewriteURBound} we may be tempted to conclude that for any theory it would be in Bob's 
best interest to perform measurements that are very compatible and have weak uncertainty relations in 
the sense that the values $\zeta_{\xstr}$ can be very large. 
But, can Alice prepare states $\sigma_{s,a}$ that
attain $\zeta_{\xstr}$ for any choice of Bob's measurements?

The other important ingredient in understanding non-local games is thus steerability. We can think
of Alice's part of the strategy $\setS,\sigma_{AB}$ as preparing the ensemble $\{p(a|s),\sigma_{s,a}\}_a$ on Bob's system, whenever she receives question $s$.
Thus when considering the optimal strategy for non-local games, 
we want to maximise $\pgame$ over all ensembles $\ens_s = \{p(a|s),\sigma_{s,a}\}_a$ that
Alice can steer to, and use the optimal measurements $\setTopt$ for Bob. That is,
\begin{align}
\pgameMAX = \max_{\{\ens_s\}_s}
\sum_s p(s) \sum_a p(a|s)\ \zetasteer(\setTopt,\qBob)\ ,
\label{eq:optimalpgame}
\end{align}
and hence the probability that Alice and Bob win the game depends only on the strength of the uncertainty relations
with respect to the sets of steerable states.
To achieve the upper bound given by Eq.~\ref{eq:gameRewriteURBound} Alice needs to be able to prepare the ensemble $\{p(a|s),\rho_{\xstr}\}_a$ of
maximally certain states on Bob's system. 
In general, it is not clear that the 
maximally certain states
for the measurements which are optimal for Bob in the game can be steered to.

It turns out that in quantum mechanics, this can be achieved in cases where
we know the optimal strategy.
For all \xor\ games~\footnote{In an XOR game the answers $a \in \01$ and $b \in \01$ of Alice and Bob respectively are single bits, and the decision whether Alice and Bob win or not depends only on the XOR of the answers $a \oplus b = a + b \mod 2$.},
that is correlation inequalities for two outcome observables 
(which include CHSH as a special case), as well as 
other games 
where the optimal measurements are known
we find that
the states which are maximally certain can be steered to (see appendix).
The uncertainty relations for Bob's optimal measurements thus give a tight bound
\begin{align}\label{eq:XORexact}
\pgameMAX = \sum_s p(s) \sum_a p(a|s) \zetamax(\setTopt,\qBob)\ ,
\end{align}
where we recall that $\zetamax$ is the bound given by the maximization over 
the full set of allowed states on Bob's system. It is an open 
question whether this holds for all games in quantum mechanics. 

An important consequence of this is that any theory that allows Alice and Bob to win with a probability exceeding
$\pgameMAX$ requires measurements which do not respect
the fine-grained uncertainty relations given by $\zetamax$ for the 
measurements used by Bob (the same
argument can be made for Alice). 
Even more, it can lead to a violation of the corresponding min-entropic uncertainty relations (see appendix).
For example, if quantum mechanics were to violate the Tsirelson bound using the same measurements for Bob,
it would need to violate the min-entropic uncertainty relations~\cite{deutsch:uncertainty}.
This relation holds even if Alice and Bob were to perform altogether different measurements when winning the game
with a probability exceeding $\pgameMAX$: for these new measurements there exist analogous uncertainty relations on the set $\Sigma$ of steerable states, 
and a higher winning probability thus always leads to a violation of such an uncertainty relation.
Conversely, 
if a theory allows any states violating even one of these fine-grained uncertainty relations for Bob's (or Alice's) optimal measurements
on the sets of steerable states, then Alice and Bob are able to violate the Tsirelson's type bound for the game. 

Although the connection between non-locality and uncertainty is more general, 
we examine the example of the CHSH inequality in detail to gain some intuition on how uncertainty
relations of various theories determine the extent to which
the theory can violate Tsirelson's bound (see appendix).
To briefly summerize, in quantum theory, Bob's optimal measurements are $X$ and $Z$ which have
uncertainty relations given by $\zetamaxx=1/2+1/(2\sqrt{2})$ of Eq. \ref{eq:uncertaintyexample}.  Thus, if Alice could steer to the maximally certain
states for these measurements, they would be able to achieve a winning
probability given by $\pgameMAX=\zetamaxx$ i.e. the degree of non-locality 
would be determined by the uncertainty relation.  This is the case -- if
Alice and Bob share the singlet state then Alice can steer to the maximally 
certain states by  measuring in the basis given by the  
eigenstates of $(X + Z)/\sqrt{2}$ or of $(X - Z)/\sqrt{2}$.
For quantum mechanics, our ability to steer is only limited by the no-signalling
principle, but we encounter strong uncertainty relations. 

On the other hand, for a local hidden variable
theory, we can also have perfect steering, but only with
an uncertainty relation given
by $\zetamaxx=3/4$, and thus we also have the degree of non-locality
given by $\pgameMAX=3/4$.  This value of non-locality
is the same as deterministic classical mechanics where we have no uncertainty relations on the full set of deterministic states, but our abilities
to steer to them are severely limited.
In the other direction, there are theories
that are maximally non-local, yet 
remain no-signalling~\cite{PR}.  These have no uncertainty, i.e. $\zetamaxx=1$, but unlike in the classical world we still have perfect
steering, so they win the CHSH game with probability $1$.

For any physical theory we can thus consider the strength of non-local correlations to be a tradeoff
between two aspects: steerability and uncertainty. 
In turn, the strength of 
non-locality can determine
the strength of uncertainty in measurements. However, it does not determine the strength of complementarity of measurements (see appendix).
The concepts of uncertainty and complementarity are usually linked, but
we find that one can have theories that are just as non-local and uncertain
as quantum mechanics, but that have less complementarity.  This suggests
a rich structure relating these quantities, which may be elucidated by further research in the direction suggested here.

\myacknowledgments
The retrieval game used was discovered in
collaboration with Andrew Doherty.
JO is supported by the Royal Society.
SW is supported by NSF grants PHY-04056720 and PHY-0803371,  
the National Research Foundation and the Ministry of Education, Singapore.
Part of this work was done while JO was visiting
Caltech (Pasadena, USA), and while JO and SW were visiting KITP (Santa Barbara, USA) funded
by NSF grant PHY-0551164.

\bibliographystyle{plain}

\appendix

\section*{Supplementary Material}

In this appendix we provide the technical details used or discussed in the main part.
In Section \ref{sec:xor} we 
show 
that for the case of XOR games, Alice can always steer to the maximally certain states of Bob's optimal
measurement operators.
That is, the relation between 
uncertainty relations and non-local games does not depend on any additional steering constraints.  
Hence, a violation of the Tsirelson's bound implies a violation of the corresponding
uncertainty relation.  
Conversely, 
a violation of 
the uncertainty relation leads to a violation of the Tsirelson bound as long 
the theory allows Alice to steer to the new maximally certain states.
The famous CHSH game is a particular example of an XOR game,
and in Sections \ref{sec:retrieval} and \ref{sec:minentropic}
we find that 
Tsirelson's bound~\cite{tsirel:original} is violated if and only if Deutsch' 
min-entropic uncertainty relation~\cite{deutsch:uncertainty} 
is violated, whenever steering is possible. 
In fact, for an even wider class of games called {\it retrieval games}
a violation of Tsirelson's bound implies 
a violation of the min-entropic uncertainty relations 
for Bob's optimal measurement derived in~\cite{ww:cliffordUR}.

In Section~\ref{sec:minentropic} we show that 
our fine-grained uncertainty relations are not only directly related to entropic uncertainty relations
\begin{align}
\sum_{t=1}^n \hmin(t)_\sigma \geq - \log \max_{\str} \zetamax\ ,
\end{align}
but they are particularly appealing from both a physical, as well as an information processing perspective:
 For measurements in a full set of
so-called mutually unbiased bases, the $\zetamax$ are the extrema of the discrete Wigner function used
in physics~\cite{wootters:wigner}. 
From an information processing perspective, we may think of the string $\str$ as being encoded into
quantum states $\sigma_{\str}$ where we can perform the measurement $t$ to retrieve the $t$-th entry from $\str$. Such an encoding is
known as a random access encoding of the string $\str$~\cite{nayak:original,nayak:rac}. 
The value of $\zeta_\str$ can now be understood as an upper bound
on the probability
that we retrieve the desired entry correctly, averaged over our choice of $t$ which is just $\psucRAC(\sigma_{\str})$. This upper bound
is attained by the maximally certain state $\rho_{\vec{x}}$.
Bounds on the performance of random access encodings thus translate directly into uncertainty relations and vice versa.

In Section~\ref{sec:examples} we discuss a number of example theories in the context of CHSH
in order to demonstrate the interplay between uncertainty and non-locality.
This includes quantum mechanics, classical mechanics, local hidden variable models and theories
which are maximally non-local but still non-signalling. 
In Section \ref{sec:compl}, we show that although
non-locality might determine the extent to which measurements
in a theory are uncertain, it does
not determine how complementary the measurements are.  We distinguish
these concepts and give an example
in the case of the CHSH inequality, of a theory which is just as non-local
and uncertain as quantum mechanics, but where the measurements are less
complementary.
Finally, in Section \ref{sec:URtoNL} we show how the form of any uncertainty relation can be used to construct 
a non-local game. 
\appendix
\section{XOR-games}
\label{sec:xor}
In this section, we concentrate on showing that for the case of XOR-games, the relation between 
uncertainty relations and non-local games does not involve steering constraints, since the steering is only limited by the no-signalling condition.  Quantum mechanics
allows Alice to steer Bob's states to those which are maximally certain for his
optimal measurement settings.
For quantum mechanics to become more non-local, it must have weaker
uncertainty relations for steerable states. 

First of all, let us recall some important facts about \xor\ games. 
These games are equivalent to Bell inequalities for observables with two outcomes i.e.  
bipartite correlation inequalities for dichotomic observables.  They
form the only general class of games which are truly understood at present. Not only do we know the structure of the measurements that 
Alice and Bob will perform,
but we can also find them efficiently using semidefinite programming for any XOR game~\cite{wehner05d}, 
which remains a daunting problem for general games~\cite{as:SDPhierarchy, acin:bell,acin:bell2}.

\subsection{Structure of measurements}

We first recall the structure of the optimal measurements used by Alice and Bob.
In an XOR game, the answers $a \in \01$ and $b \in \01$ of Alice and Bob respectively are single bits, and the decision whether Alice and Bob win or not depends only on the XOR
of the answers $a \oplus b = a + b \mod 2$. 
The winning condition can thus be written in terms of a predicate that obeys
$V(c = a\oplus b|s,t) = 1$ if and only if $a \oplus b = c$ are winning answers for Alice and Bob given 
settings $s$ and $t$ (otherwise $V(c|s,t) = 0$).
Let $A_s^a$ and $B_t^b$ denote the measurement operators corresponding to measurement settings $s$ and $t$ giving outcomes
$a$ and $b$ of Alice and Bob respectively. Without loss of generality, we may thereby assume that these are projective measurements
satisfying $A_s^a A_s^{a'} = \delta_{a,a'} A_s^a$, and similarly for Bob.
Since we have only two measurement outcomes, we may view this measurement as measuring the observables
\begin{align}
A_s &=A_s^0 - A_s^1\ , \label{eq:AliceObservable}\\
B_t &=B_t^0 - B_t^1\ , \label{eq:BobObservable}
\end{align}
with outcomes $\pm 1$ where we label the outcome '$+1$' as '$0$', and the outcome '$-1$' as '$1$'.
Tsirelson~\cite{tsirel:original,tsirel:separated} has shown that the optimal winning probability for Alice and Bob 
can be achieved using traceless observables of the form
\begin{align}
A_s &= \sum_{j} a_s^{(j)} \Gamma_j\ ,\label{eq:AliceXORoperators}\\
B_t &= \sum_{j} b_t^{(j)} \Gamma_j\ ,\label{eq:BobXORoperators}
\end{align}
where $\vec{a}_s = (a_s^{(1)},\ldots,a_s^{(N)}) \in \Real^N$ and $\vec{b}_t = (b_t^{(1)},\ldots,b_t^{(N)}) \in \Real^N$ 
are real unit vectors of dimension $N = \min(|\setS|,|\setT|)$ and $\Gamma_1,\ldots,\Gamma_N$ are the anti-commuting generators of a Clifford
algebra. That is, we have
\begin{align}
\{\Gamma_j,\Gamma_k\} &= \Gamma_j \Gamma_k + \Gamma_k \Gamma_j = 0 \mbox{ for all } j\neq k\ ,\label{eq:cliffordAntiComm}\\
(\Gamma_j)^2 &= \id \mbox{ for all } j\label{eq:cliffordNormalized}
\end{align}
In dimension $d = 2^{n}$ we can find at most $N = 2n+1$ such operators, which can be obtained using the well known Jordan-Wigner transform~\cite{JordanWigner}.
Since we have 
$A_s^0 + A_s^1 = \id$ and
$B_t^0 + B_t^1 = \id$
we may now use Eqs.~\ref{eq:AliceObservable} and~\ref{eq:BobObservable} 
to express the measurement operators in terms of the observables as 
\begin{align}
A_s^a & = \frac{1}{2}\left(\id + (-1)^a A_s\right)\ ,\label{eq:AliceMeasureOps}\\
B_t^b &= \frac{1}{2}\left(\id + (-1)^b B_t\right)\ .\label{eq:BobMeasureOps}
\end{align}
Tsirelson furthermore showed that the optimal state that Alice and Bob use with said observables is
the maximally entangled state
\begin{align}
\ket{\psi} = \frac{1}{\sqrt{d}}
\sum_{k=0}^{d-1} \ket{k}\ket{k}
\end{align}
of total dimension $d = (2^{\floor{N/2}})^2$.

To see how XOR games are related to Bell correlation inequalities note that we may use Eqs.~\ref{eq:AliceMeasureOps} and~\ref{eq:BobMeasureOps}
to rewrite the winning probability of the game in terms of the observables as
\begin{align}
\pgame &= \sum_{s,t} p(s,t) \sum_{a,b} V(a,b|s,t) \bra{\psi}A_s^a \otimes B_t^b\ket{\psi}\\
&= \frac{1}{2}
\left(\sum_{s,t} p(s,t) \sum_c V(c|s,t)
\left(1 + (-1)^c \bra{\psi}A_s \otimes B_t\ket{\psi}\right)\right)\ .
\end{align}
For example, for the well known CHSH inequality given by
\begin{align}
\bra{\psi}CHSH \ket{\psi}
&= \bra{\psi} A_0 \otimes B_0 + A_0 \otimes B_1 + A_1 \otimes B_0 - A_1 \otimes B_1\ket{\psi}\ 
\end{align}
with $p(s,t) = p(s) p(t) = 1/4$ we have 
\begin{align}
\pgame &= \frac{1}{2} \left(1 + \frac{\bra{\psi} CHSH \ket{\psi} }{4}\right)\\
&= \frac{1}{4} \sum_{s,t} \sum_{a,b} V(a,b|s,t) \bra{\psi} A_s^a \otimes B_t^b \ket{\psi}\ ,
\end{align}
where $V(a,b|s,t) = 1$ if and only if $a \oplus b = s \cdot t$.

\subsection{Steering to the maximally certain states}

Now that we know the structure of the measurements that achieve the maximal winning probability in an XOR game, we can
write down the corresponding uncertainty operators for Bob 
(the case for Alice is analogous). Once we have done that,
we will see that it is possible for Alice to steer Bob's state to the maximally certain ones.
First of all, note that in the quantum case
we have that the set of measurement operators
$\{B_t\}_t$ is each decomposable in terms of POVM elements $\{B_t^b\}_b$, 
and thus
$\zetamax$
corresponds to the largest eigenvalue of the \emph{uncertainty operator}
\begin{align}
\unop_{\xstr}=\sum_{t} \prob{t}  B_t^{x_{s,a}^{(t)}}\ .
\label{eq:qunop}
\end{align}
The $\unop$ are positive operators, so sets of them, together with an error
operator form a positive valued operator
and can be measured.
In the case of XOR games, we thus have
\begin{align}
\unop_{\xstr}=\sum_{t} \prob{t} \sum_{\substack{b\\ V(a,b|s,t) = 1}} B_t^b\ ,
\label{eq:unop}
\end{align}
and the Bell polynomial can be expressed as
\begin{align}
\beta &= \sum_{stab}\prob{s}\prob{t}V(ab|st) A_s^a\otimes B_t^b\\ 
&= \sum_{sa}  \prob{s}A_s^a\otimes \unop_{\xstr}\ .
\end{align}
For the special case of XOR-games we can now use Eq.~\ref{eq:BobMeasureOps} to rewrite Bob's uncertainty operator using the short
hand 
$\vec{v} \cdot \vec{\Gamma} = \sum_j v^{(j)} \Gamma_j$ as
\begin{align}
\unop_{\xstr} &= \frac{1}{2} \sum_t \prob{t} \sum_{\substack{b\\V(a,b|s,t) = 1}} \left(\id + (-1)^b B_t\right)\nonumber\\
&= \frac{1}{2}\left[c_{s,a} \id + \vec{v}_{s,a} \cdot \vec{\Gamma}\right]
\label{eq:uncertRewrite}\ ,
\end{align}
where 
\begin{align}
c_{s,a} &= \sum_t \prob{t} |\sum_b V(a,b|s,t)|\, \\
\vec{v}_{s,a}  &= \sum_t \prob{t} \sum_{\substack{b\\V(a,b|s,t) = 1}}
(-1)^b\ \vec{b}_t
\end{align}

What makes XOR games so appealing, is that we can now easily compute the maximally certain states. 
Except for the case of two measurements, and rank one measurement operators
this is generally not so easy to do, even though bounds may of course be found. Finding a general expression for such maximally certainty states 
allows us 
easily to compute
the extrema of the discrete Wigner function as well as
to prove
tight entropic uncertainty relations for all measurements.
This has been an open problem since Deutsch's initial bound for two measurements 
using rank one operators~\cite{deutsch:uncertainty,ww:URsurvey}.
In particular, we recall for completeness~\cite{ds:pistar} that

\begin{claim}
\label{claim:uncertaintyVectors}
For any XOR game, we have that for any uncertainty operator $\unop_{\xstr}$ of observables of the 
form Eqs.~\ref{eq:AliceXORoperators} and~\ref{eq:BobXORoperators}
\begin{align}
\zetamaxx =
\max_{\substack{\rho \geq 0\\\tr{\rho} = 1}} \tr\left(\rho \unop_{\xstr}\right) = \tr(\rho_{\xstr} \unop_{\xstr}) = \frac{1}{2}\left(c_{s,a} + \|\vec{v}_{s,a}\|_2\right)
\end{align}
with
\begin{align}
\rho_{\xstr} = \frac{1}{d}\left(\id + \sum_j r_{s,a}^{(j)} \Gamma_j\right)
\label{eq:XORoptimalUncertState}
\end{align}
and
\begin{align}
\vec{r}_{s,a} = \vec{v}_{s,a}/\|\vec{v}_{s,a}\|_2\ .
\end{align}
\end{claim}
\begin{proof}
We first note that the set of operators $\{\id\}\cup \{\Gamma_j\}_j \cup \{i \Gamma_j \Gamma_k\}_{jk} \cup \ldots$ forms an orthonormal~\footnote{With respect to the Hilbert-Schmidt inner product.}
basis for $d \times d$ Hermitian operators~\cite{dietz:bloch}. We may hence write any state $\rho$ as
\begin{align}
\rho = \frac{1}{d}\left(\id + \sum_j r^{(j)}_{s,a} \Gamma_j + \ldots \right)\ ,
\end{align}
for some real coefficients $r^{(j)}_{s,a}$. Using the fact that $\mathcal{S}$ forms an orthonormal basis, we then obtain from Eq.~\ref{eq:uncertRewrite} that
\begin{align}
\tr(\rho \unop_{\xstr}) = \frac{1}{2}\left(c_{s,a} + \vec{r}_{s,a} \cdot \vec{v}_{s,a}\right)\ .
\end{align}
Since for any quantum state $\rho$ we must have $\|\vec{r}_{s,a}\|_2 \leq 1$~\cite{ww:cliffordUR},
maximising the left hand side corresponds to maximising the right hand side over vectors $\vec{r}_{s,a}$ of unit length.
The claim now follows by noting that if $\vec{r}_{s,a}$ has unit length, then $\rho_{\xstr}$ is also a valid
quantum state~\cite{ww:cliffordUR}.
\end{proof}
Note that the states $\rho_{\xstr}$ are highly mixed outside of dimension $d=2$, but still correspond to the maximally certain states.
We now claim that that for any setting $s$, Alice can in fact steer to the states which are optimal for Bob's uncertainty operator.

\begin{lemma}
For any XOR game we have for all $s,\hat{s} \in \setS$
\begin{align}
\sum_a p(a|s)\ \rho_{\xstr} = \sum_a p(a|\hat{s})\ \rho_{\hat{s},a}\ ,
\end{align}
where $\rho_{\xstr}$ is defined as in Eq.~\ref{eq:XORoptimalUncertState}.
\end{lemma}
\begin{proof}
First of all, note that it is easy to see in any XOR game
the marginals of Alice (and analogously, Bob) are uniform. That is, the probability that
Alice obtains outcome $a$ given setting $s$ obeys
\begin{align}
\prob{a|s} &= \bra{\psi} A_s^a \otimes \id\ket{\psi}\\
&= \frac{1}{2} \left(\bra{\psi}\id \otimes \id\ket{\psi} + (-1)^a \bra{\psi}A_s \otimes \id\ket{\psi}\right)\nonumber\\
&=\frac{1}{2}\left(1 + \tr(A_s^T)\right) = \frac{1}{2}\nonumber \ .
\end{align}
Second, it follows from the fact that we are considering an XOR game that $\vec{v}_{s,a} = - \vec{v}_{s,1 - a}$ for all $s$ and $a$.
Hence, by Claim~\ref{claim:uncertaintyVectors} we have $\vec{r}_{s,a} = - \vec{r}_{s,1 - a}$ and hence
\begin{align}
\frac{1}{2}\left(\rho_{{\vec{x}_{s,0}}} + \rho_{{\vec{x}_{s,1}}}\right) = \frac{\id}{d}\ ,
\end{align}
from which the statement follows.
\end{proof}

Unrelated to our problem, it is interesting to note that the observables $A_s = \vec{a}_s \cdot \vec{\Gamma}$ employed by Alice are determined exactly by the vectors
$r_{s,a}$ above. Let $\sigma_{a,s}$ denote Bob's post-measurement state which are steered to when Alice performs 
the measurement labelled by $s$ and obtains outcome $a$.
We then have that for any of Bob's measurement operators
\begin{align}
\tr\left(\sigma_{s,a} \unop_{\xstr}\right) &= \frac{\bra{\psi}A_s^a \otimes \unop_{\xstr}\ket{\psi}}{p(a|s)}\\
&= \frac{1}{2}\left(c_{s,a} + (-1)^a \vec{a}_{s} \cdot \vec{v}_{s,a}\right)\ .
\end{align}
Since Alice's observables should obey $(A_s)^2 = \id$, we would like that $\|\vec{a}_s\|_2= 1$. When trying to find the optimal measurements
for Alice we are hence faced with exactly the same task as when trying to optimize over states $\rho_{\xstr}$ in Claim~\ref{claim:uncertaintyVectors}.
That is, letting $\vec{a}_s = (-1)^a r_{s,a}$ gives us the optimal observables for Alice.

\subsection{Retrieval games and entropic uncertainty relations}
\label{sec:retrieval}

As a further example, we now consider a special class of non-local games that we call 
a {\it retrieval games}, which are a class of XOR games. This class includes the famous CHSH game.
Retrieval games have the further property, that not only are the fine-grained uncertainty relations violated if quantum mechanics would be more non-local, but also the min-entropic uncertainty relation
derived in \cite{ww:cliffordUR} would be violated.  The class of retrieval
games are those which correspond directly to the retrieval of a bit from an
$n$-bit string
More precisely, we label the settings of Alice
using all the $n$-bit strings starting with a '0'. That is,
\begin{align}
\setS = \{ (0,s^{(2)},\ldots,s^{(n)}) \mid s^{(j)} \in \01\}\ ,
\end{align}
where we will choose each setting with uniform probability $p(s) = 1/2^{n-1}$.
For each setting $s$, Alice has two outcomes which we label using the strings $s$ and $\underbar{s}$, where $\underbar{s}$ is the bitwise complement of the
string $s$. More specifically, we let 
\begin{align}
\str_{s,0} &= s\ ,\\
\str_{s,1} &= \underbar{{\it{s}}}\ . 
\end{align}
Bob's settings are labelled by the indices $\setT = \{1,\ldots,n\}$, and chosen uniformly at random $p(t) = 1/n$. 
Alice and Bob win the game if and only if Bob outputs the bit at position $t$ of the string
that Alice outputs. In terms of the predicate, the game rules can be written as
\begin{align}
V(a,b|s,t) = \left\{\begin{array}{cc}
1 & \mbox{ if } b = \str_{s,a}^{(t)}\ ,\\
0 & \mbox{ otherwise} \ .
\end{array}
\right.
\end{align}
We call such a game a \emph{retrieval game} of length $n$. It is not hard to see that the CHSH game above is a retrieval game of length
$n=2$~\cite{as:dimBound}.

We now show that for the case of retrieval games, not only can Alice steer perfectly to the optimal uncertainty states, but
Bob's optimal measurements are in fact very incompatible in the sense that his measurement operators necessarily anti-commute.
In particular, this means that we obtain very strong uncertainty relations
 for Bob's optimal measurements~\cite{ww:cliffordUR},
as we will discuss in more detail below.

\begin{lemma}
Let $B_t = B_t^0 - B_t^1$ denote Bob's optimal dichotomic observables in dimension $d = 2^{n}$. 
Then for any retrieval game
\begin{align}
\{B_t,B_{t'}\} = 0 \mbox{ for all } t \neq t' \in \setT\ .
\end{align}
\end{lemma}
\begin{proof}
First of all note that
using Claim~\ref{claim:uncertaintyVectors} we can write the winning probability for 
arbitrary observables $B_t = \vec{b}_t \cdot \vec{\Gamma}$
in terms of the maximum uncertainty states $\rho_{\xstr}$ as
\begin{align}
\pgame &= \frac{1}{2^n} \sum_{s,a} \tr\left(\rho_{\xstr} \unop_{\xstr}\right)\\
&= \frac{1}{2^n} \sum_{s,a} \frac{1}{2}\left(1  + \|\vec{v}_{s,a}\|_2\right)\\
& = \frac{1}{2}\left(1  + \frac{1}{2^n} \sum_{s,a} \sqrt{\vec{v}_{s,a} \cdot \vec{v}_{s,a}}\right)\\
&\leq \frac{1}{2}\left(1 + \sqrt{\frac{1}{2^n} \sum_{s,a} \sum_{t t'} (-1)^{\str_{s,a}^{(t)}} (-1)^{\str_{s,a}^{(t')}} \vec{b}_t \cdot \vec{b}_{t'}}\right)\\
& = \frac{1}{2} + \frac{1}{2 \sqrt{n}}\ ,
\end{align}
where the inequality follows from the concavity of the square-root function and Jensen's inequality, and the final inequality from the fact that
$\vec{b}_t$ is of unit length.
Equality with this upper bound can be achieved by choosing $\vec{b}_t \cdot \vec{b}_{t'} = 0$ for all $t \neq t'$, which is equivalent to
$\{B_t,B_{t'}\} = 0$ for all $t \neq t'$.
\end{proof}

For our discussions about uncertainty relations in the next section, it will be useful to note that the above implies that
$\lambda_{\rm max}(\unop_{\xstr}) = \lambda_{\rm max}(\unop_{{\vec{x}_{s',a'}}})$ for all settings $s,s' \in \setS$ and outcomes $a,a' \in \setA$.
A useful consequence of the fact that $\vec{b}_t \cdot \vec{b}_{t'} = 0$ for $t \neq t'$ is also that
the probability that Bob retrieves
a bit of the string $\xstr$ correctly is the same for all bits:

\begin{corollary}\label{cor:equalForBits}
For a retrieval game, we have for the optimal strategy that for all settings $s \in \setS$ and outcomes
$a \in \setA$
\begin{align}
\tr\left(\rho_{\xstr} B_t^{\str_{s,a}^{(t)}}\right) = 
\tr\left(\rho_{\xstr} B_{t'}^{\str_{s,a}^{(t')}}\right)  \mbox{ for all } t,t' \in \setT\ .
\end{align}
\end{corollary}

The fact that min-entropic uncertainty relations are violated if quantum mechanics were to do better than the Tsirelson's bound for retrieval games will be shown in 
Section \ref{sec:minentropic}.

It is an interesting open question, whether Alice can steer to the maximally certain states for Bob's optimal measurements
for all games. This is for example possible, for the $\mod 3$-game suggested in~\cite{massar:game} to which the optimal measurements
were derived in~\cite{ji:game,yeongCherng:game}. Unfortunately, the structure of measurement operators that are optimal for Bob
is ill-understood for most games, which makes this a more difficult task.

\section{Min-entropic uncertainty relations}
\label{sec:minentropic}

We now discuss the relationship between fine-grained uncertainty relations
and min-entropic ones. The relation between retrieval games and min-entropic uncertainty relations will follow from that.
Recall that the min-entropy of the distribution obtained by measuring a state $\rho$ using the measurement
given by the observables $B_t$ can be written as
\begin{align}
\hmin(B_t)_\rho = - \log \max_{x^{(t)} \in \01} \tr(B_t^{x^{(t)}} \rho)\ .
\end{align}
A min-entropic uncertainty relation for the measurements in an $n$-bit retrieval game can be bounded as
\begin{align}
\frac{1}{n} \sum_{t = 1}^{n} \hmin(B_t)_\rho &= - \frac{1}{n} \sum_{t} \log \max_{x^{(t)} \in \01} 
\tr\left(B_t^{x^{(t)}} \rho\right)\\
&\geq - \log \max_{\vec{\tilde{x}} \in \setB^{\times n}} \sum_t \tr\left(B_t^{\tilde{x}^{(t)}} \rho\right)\label{eq:minEntropyJensen}\\
&= - \log \max_{s,a} \lambda_{\rm max}(\unop_{\xstr})\ ,
\end{align}
where the inequality follows from Jensen's inequality and the concavity of the log. I.e. the fine-grained relations provide a lower bound on the min-entropic
uncertainty relations.

Now we note that for the case of retrieval games, not only do we have
\begin{align}
\lambda_{\rm max}(\unop_{\xstr}) = \frac{1}{2} + \frac{1}{2\sqrt{n}}\ ,
\end{align}
but by Corollary~\ref{cor:equalForBits} we have that the inequality~\ref{eq:minEntropyJensen} is in fact tight, where
equality is achieved for any of the maximum uncertainty states $\rho_{\xstr}$ from the retrieval game.  
We thus have

\begin{theorem}\label{eq:retrievalMinEntropy}
For any retrieval game, a violation of the Tsirelson's bound implies
 a violation of the min-entropic uncertainty relations for Bob's optimal measurements.  A violation of the min-entropic uncertainty relation implies a violation
of the Tsirelson's bound as long as steering is possible. 
\end{theorem}

Note that since the CHSH game is a retrieval game with $n=2$, we have that as long as steering is possible, 
Tsirelson's bound~\cite{tsirel:original} is violated if and only if Deutsch' min-entropic uncertainty relation~\cite{deutsch:uncertainty} 
is violated for Bob's (or Alice's) measurements.

\section{An example: the CHSH inequality in general theories}\label{sec:examples}

Probably the most well studied Bell inequality is the CHSH inequality
and previous attempts to understand the strength of quantum non-locality
have been with respect to it~\cite{ infoCausality,wim:nonlocal,brassard2006limit,allcock2009recovering}.
Although the connections between non-locality and uncertainty are more general, 
we can use the CHSH inequality as an example and show how the uncertainty
relations of various theories determine the extent to which
the theory can violate it.  We will see in this example how non-locality requires steering (which is
what prevents classical mechanics from violating a Bell inequality despite having maximal 
certainty).  Furthermore, we will see that local-hidden variable theories can have increased steering ability, but don't violate
a Bell inequality because they exactly compensate by having more uncertainty.  Quantum mechanics has perfect
steering, and so it's non-locality is limited only by the uncertainty principle.  We will also discuss
theories which have the same degree of steering as quantum theory, but greater non-locality because 
they have greater certainty (so-called ``PR-boxes''~\cite{PR,PR1,PR2} being an example).

The CHSH inequality can be expressed as a game
in which Alice and Bob receive
binary questions $s, t \in \01$ respectively, and similarly their answers
$a, b\in \01$ are single bits.  Alice and Bob win the CHSH game if their answers satisfy
$a \oplus b = s \cdot t$. The CHSH game thus belongs to the class of XOR games, and any other XOR game could be used as a similar example.

Note that we may again rephrase the game in the language of random access coding~\cite{as:dimBound}, where
we label Alice's outcomes using string $\xstr$ and Bob's goal is to retrieve the $t$-th element of this string. 
For $s=0$, Bob will always need to give the same answer as Alice in order to win independent of $t$, and hence we have
$\vec{x}_{0,0} = (0,0)$, and $\vec{x}_{0,1} = (1,1)$. For $s=1$, Bob needs to give the same answer for $t=0$, but the opposite
answer if $t=1$. That is, $\vec{x}_{1,0} = (0,1)$, and $\vec{x}_{1,1} = (1,0)$.

To gain some intuition of the tradeoff between steerability and uncertainty, we consider an (over)simplified example
in Figure~\ref{fig:tradeoff}. We examine quantum, classical, and a theory allowing maximal non-locality below.

\begin{figure}
\begin{center}
\includegraphics[scale=.85]{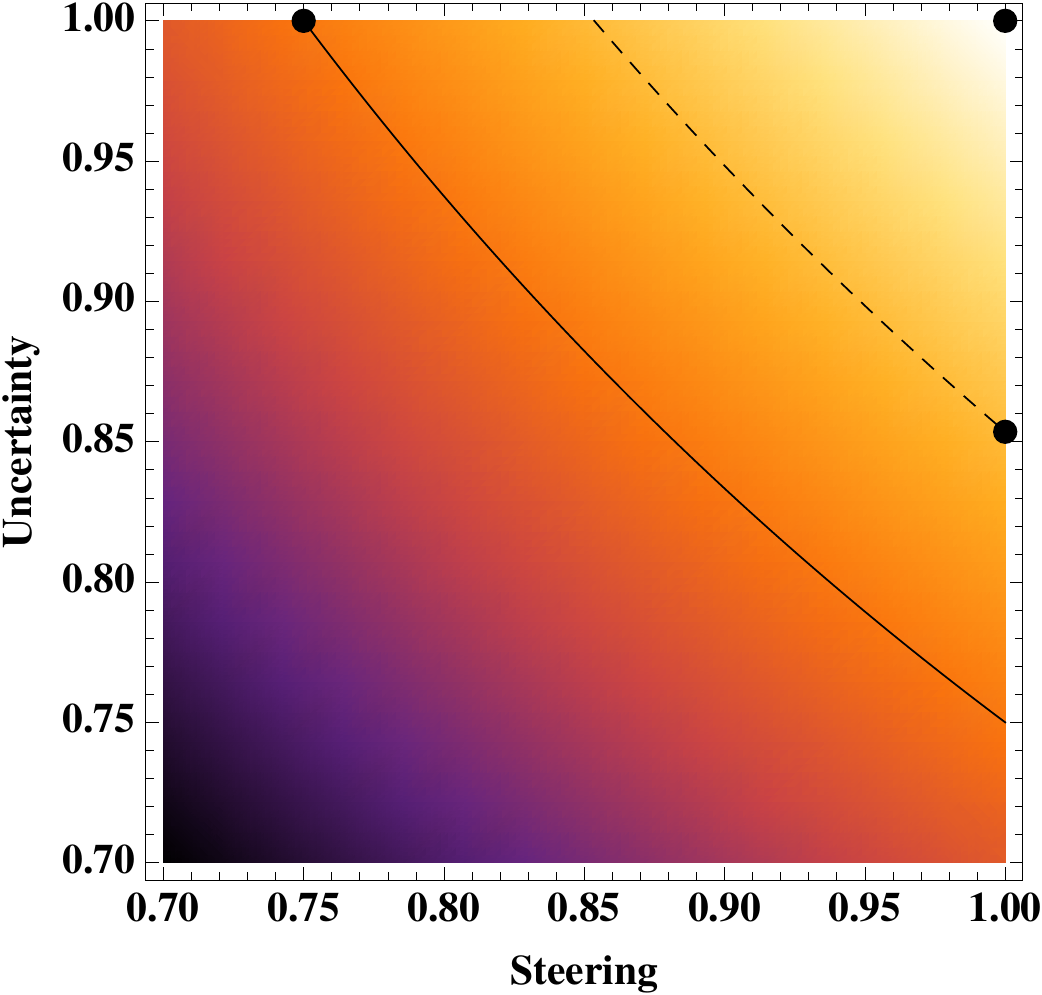}
\caption{A simplified example: Imagine a world in which the only steerable states are the maximally certain states of the uncertainty relations we consider,
and we have an all or nothing form of steering. I.e., either Alice can steer to all ensembles with probability $p_{\rm steer}$, or fails entirely.
The vertical axis denotes the certainty $p_{\rm cert}$ (that is, the lack of \emph{uncertainty}), and the horizontal axis $p_{\rm steer}$.
Lighter colours indicate a larger winning probability, which in this simplified case is just $\pgame = p_{\rm steer} p_{\rm cert}$.
The solid line denotes the case of $\pgameMAX = 3/4$, which can be achieved classically.
The point on the line denotes the combination of values for a classical deterministic theory: there is no uncertainty
($\zeta_{\xstr} = 1$ for all $\xstr$), and no steering other than the trivial one to the state Alice and Bob already share as part of their strategy which yields $3/4$ on average. The dashed line denotes the value $\pgameMAX = 1/2 + 1/(2\sqrt{2})$ achievable by a quantum strategy. The point on the line denotes the point reached quantumly: there is uncertainty, but we can steer perfectly to the maximally certain states. Finally, the point at $(1,1)$ denotes the point achievable by ``PR-boxes'': there is no uncertainty, but nevertheless perfect steering.
}
\label{fig:tradeoff}
\end{center}
\end{figure}

\noindent
{\bf (i) Quantum mechanics}:

As we showed in Section \ref{sec:xor} 
we have for any~\xor\ game that we can always steer to the maximally certain states 
$\rho_{\xstr}$, and hence Alice's and Bob's winning probability depend only on the uncertainty relations.
For CHSH, Bob's optimal measurement are given by the binary observables $B_0 = Z$ and $B_1 = X$.
The amount of uncertainty we observe for Bob's optimal measurements is given by
\begin{align}
\zetamaxx
= \frac{1}{2} + \frac{1}{2\sqrt{2}} \mbox{ for all } \xstr \in \01^2\ ,
\end{align}
where the maximally certain states are given by the eigenstates of $(X+Z)/\sqrt{2}$ and
$(X - Z)/\sqrt{2}$. Alice can steer Bob's states to the eigenstates of these
operators by measuring in the basis of these operators on the state
\begin{align}
\ket{\psi^-}=\frac{1}{\sqrt{2}}(\ket{00}_{AB}+\ket{11}_{AB})\ .
\end{align}
We hence have $\pgameMAX = \zetamaxx = 1/2 + 1/(2\sqrt{2})$ 
which is Tsirelson's 
bound~\cite{tsirel:original,tsirel:separated}. If Alice and Bob
could obtain a higher value for the same measurements, at least one of the fine-grained uncertainty
relations is violated. We also saw above that a larger violation
for CHSH also implies a violation
of Deutsch' min-entropic uncertainty relation.

\noindent
{\bf (ii) Classical mechanics \& local hidden variable theories}:

For classical theories, let us first consider the case where we use a deterministic strategy,
since classically there is no fundamental restriction on how much 
information can be gained i.e., if we optimize the uncertainty
relations over all classical states, we have 
$\zetamaxx= 1$
for any set of measurements. 
However, for the states which are maximally certain, there is no steering property
either because in the no-signalling constraint
the density matrix has only one term in it.
A deterministic state cannot be written as a convex sum of any other states -- it is an 
extremal point of a convex set.  The best deterministic strategy
is thus to prepare a particular state at Bob's site (e.g. an encoding of the bit-string $00$).
This results in $\pgame=3/4$. A probabilistic strategy cannot do better, since it would 
be a convex combination of deterministic strategies and one might as well choose the best one.
However, it will be instructive to consider the non-deterministic case.

Although the maximally certain states cannot be steered to,
we can steer to states
which are mixtures of deterministic states (c.f.~\cite{spekkens2004defense}).  
This corresponds to using a non-deterministic
strategy, or if the non-determinism is fundamental, to 
a local hidden variable theory.
However, measurements on the states which can be steered to will not have well-defined
outcomes.  If we optimize the uncertainty relations 
with respect to the non-deterministic states, we will
find that there is a substantial uncertainty in measurement outcomes on these states.
We will find that the ability to steer is exactly compensated by
an inability to obtain certain measurement outcomes, thus the probability of winning
the non-local game will again be $\pgame=3/4$.
  
\begin{figure}
 \centering
  \includegraphics[width=3.5cm]{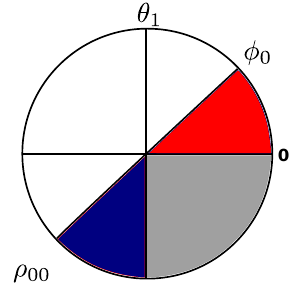}
  \includegraphics[width=3.5cm]{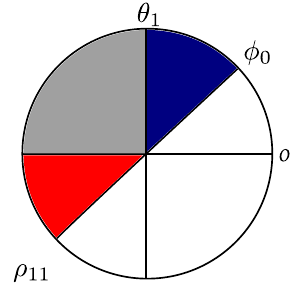}
  \includegraphics[width=3.5cm]{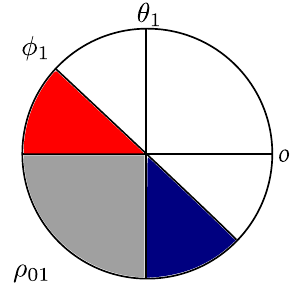}
  \includegraphics[width=3.5cm]{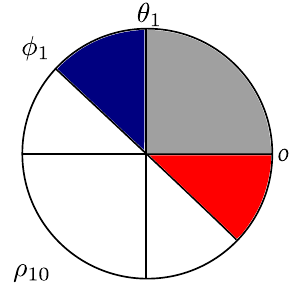}
  \caption{
The hidden variable states $\rho_{\str}$ which encode each of the four random access strings on Bob's site are depicted above.  The grey region corresponds
to the hidden variable being such that Bob can correctly retrieve both bits.
The blue and red regions correspond to the region where Bob incorrectly
retrieves the first or second bit respectively.  These regions must be included
if Alice is to be able to steer to the states.
}
 \label{fig:lhv}
\end{figure}

Figure \ref{fig:lhv}
depicts an optimal local hidden variable theory for CHSH 
where the local hidden variable is
a point on the unit circle labelled by an angle $\Omega$
and the states are probability distributions
\begin{align}
\rho_{\xstr}=\int_{\Omega_{\xstr}}^{\Omega_{\xstr}'} p(\Omega)\sigma_\Omega d\Omega
\label{eq:hiddenstate}
\end{align}
with $\sigma_\Omega$ denoting the state of the hidden
variable, and $p(\Omega)$ uniform.    Alice can prepare these states
at Bob's site if they initially share the maximally correlated state 
$\psi^{MC}_{AB}=\int_{0}^{2\pi} p(\Omega)\sigma^\Omega_A \sigma^\Omega_B d\Omega$
and she makes a partial measurement on her state which determines that the value of her hidden variable lies within
$\Omega_{\xstr}$ and $\Omega_{\xstr}'$.

A measurement by Bob corresponds 
to cutting the circle in half along some angle, and then determining whether the hidden variable lies
above or below the cut.  I.e. a coarse grained determination of the bounds of integration
in  $\rho_{\xstr}$ to within $\pi$.
Bob's measurement for $t=0$ thereby corresponds to determining whether the hidden variable lies above or below the cut along the equator,
while his measurement for $t=1$ corresponds to determining whether the hidden variable lies above or below the cut labelled by the angle $\theta_1$.  
If the hidden variable is above the cut,
we label the outcome as $1$, and $0$ otherwise.

First of all, note that if the hidden variable
lies in the region shaded in grey, as depicted in Figure \ref{fig:lhv},
then Bob will be able to retrieve both bits correctly
because it lies in the region such that the results for Bob's measurements would
match the string that the state encodes. 
For example, for the 
$\rho_{00}$ state, the hidden variable in the grey region lies below both cuts, 
while for the $\rho_{11}$ state, the hidden
variable lies above both cuts.  If we optimize the uncertainty relations over
states which are chosen from the grey region, then $\zeta_{\xstr}=1$.  If $\rho_{\xstr}$
only includes the grey
region, then it is a maximally certain state.

Whether steering is possible to a particular set of 
$\rho_{\xstr}$ is determined by the no-signalling condition.
If the bounds of integration $\Omega_{\xstr}$ and $\Omega_{\xstr}'$ in 
Equation (\ref{eq:hiddenstate})
only
included the grey regions, then we would not be able to steer to the states $\rho_{\xstr}$ 
since $p(00)\rho_{00}+p(11)\rho_{11}\neq p(01)\rho_{01}+p(10)\rho_{10}$.
That is, steering to the maximally certain states is forbidden by the no-signalling principle if they lie in the grey region.
The states only become steerable
if the $\rho_{\xstr}$ include the red and blue
areas depicted in Figure \ref{fig:lhv}. Then Alice making measurements on her share of $\psi^{MC}_{AB}$
which only determine the hidden variable
to within an angle $\pi$ will be able to prepare the
appropriate states on Bob's site.
Alice's measurement angles are 
labelled by the angles $\phi_0$ for the $s=0$ partition and $\phi_1$ for the $s=1$
partition.  

However, if the hidden variable lies in the blue area then Bob retrieves the first bit incorrectly, and if it is in the red
region he retrieves the second bit incorrectly.  
If the $\rho_{\xstr}$ include a convex combination of hidden variables which include
the grey, blue and red regions then the states are now steerable, but
the uncertainty
relations with respect to these states give
$\zeta_{00}=\zeta_{11}=1-\theta_B/2\pi$ and $\zeta_{01}=\zeta_{10}=1-(\pi-\theta_B)/2\pi$
since the error probabilities are just proportional to the size of the red and blue regions.
Hence, we obtain $\pgame=\zetamaxx=3/4$.  

\noindent
{\bf (iii) No-signalling theories with maximal non-locality}:

It is possible to construct a probability distribution which 
not only obeys the no-signalling constraint and violates the CHSH
inequality~\cite{rastall,tsirelson1993some}, but also violates it more strongly than
quantum theory, and in fact is maximally non-local~\cite{PR}.  
Objects which have this property we call PR-boxes, and in particular they
allow us to win the CHSH game with probability $1$.
Note that this implies 
that there is no uncertainty in measurement outcomes:
$\zetamaxx= 1$ for all $\xstr \in \01^2$ 
where $\rho_{\xstr}$ is the maximally certain state for
$\xstr$~\cite{barrett:nonlocal,gs:relaxedUR}.
Conditional on Alice's
measurement setting and outcome, Bob's answer must be correct for either
of his measurements labelled $t=0$ and $t=1$
and described 
by the following probability distributions
$p(b|t)$
for measurements $t \in \01$,
\begin{align}
\rho_{00} &= \{p(0|0) = 1, p(0|1) = 1\}\ ,\\
\rho_{01} &= \{p(1|0) = 1, p(1|1) = 1\}\ ,\\
\rho_{10} &= \{p(0|0) = 1, p(1|1) = 1\}\ ,\\
\rho_{11} &= \{p(1|0) = 1, p(0|1) = 1\}\ .
\end{align}
The only constraint that needs to be imposed on our ability to steer to
these states is given by the no-signalling condition, and indeed 
Alice may steer to the ensembles 
$\{1/2,\rho_{{\vec{x}_{0,a}}}\}_a$ and $\{1/2,\rho_{{\vec{x}_{1,a}}}\}_a$ at will
while still satisfying the no-signalling condition.
We hence have $\pgame =  \zetamaxx = 1$.

\section{Uncertainty and complementarity}
\label{sec:compl}

Uncertainty and complementarity are often conflated. However, even though they are closely related, they are nevertheless distinct concepts.
Here we provide a simple example that illustrates their differences; a full discussion of their relationship 
is outside the scope of this work.
Recall that the uncertainty principle as used here
is about the possible probability
distribution of measurement outcomes $p(b|t_j)$ when one measurement $t_1$ is
performed on one system and another measurement $t_2$ is performed
on an identically prepared system. Complementarity on the other hand is the
notion that you can only perform one of two incompatible measurements (see for example \cite{Bohr1935} since
one measurement disturbs the possible 
measurement outcomes when both measurements are performed
on the same system in succession. We will find that although
the degree of non-locality determines how uncertain measurements
are, this is not the case for complementarity -- there are theories
which have less complementarity than quantum mechanics, but the 
same degree of non-locality and uncertainty.

Let $\post$ denote the state of the system after we performed the measurement labelled $t$
and obtained outcome $b$.
After the measurement, we are then interested in the probability distribution 
$\ppost{b'|t'}$ of obtaining outcome $b'$ when performing measurement $t'$ on the post-measurement state.
As before, we can consider the term
\begin{align}
\eta^{\vec{x}} =
 \sum_{t'} 
p(t') \ppost{b' | t'}
\end{align}
This quantity has a similar form as 
$\zetastr$, and in the 
case where a measurement is equivalent to a preparation, we clearly 
have that 
\begin{align}\label{eq:uncertComplement}
\eta^{\vec{x}}\leq
\zetastr 
\end{align}
since post-measurement state when obtaining 
outcome $b$ is just a particular preparation, while  
$\zetastr$
is a maximisation over all preparations. 

There is however a different way of looking at complementarity, which is about
the extraction of information. In this sense, 
one would say that two measurements are complementary, if the second measurement
can extract no more information about the preparation procedure than
the first measurement and visa versa. 
We refer to this
as {\it information complementarity}.
Note that quantum mechanically, this does not necessarily
have to do with whether two measurements commute. For example,
if the first measurement is a complete Von Neumann measurements, then 
all subsequent measurements gain no new information than the first one
whether they commute or otherwise. 

\subsection{Quantum mechanics could be less complementarity with the
same degree of non-locality}

We now consider a simple example that illustrates the differences between complementarity and uncertainty.
Recall from Section~\ref{sec:retrieval} that the CHSH game is an instance of a retrieval game where Bob
is challenged to retrieve either the first or second bit of a string $\xstr \in \01^2$ prepared by Alice.
For our example, we then imagine that the initial state $\sigma_\xstr$ of Bob's system corresponds to an encoding of a string $\xstr \in \01^2$, 
and fix Bob's measurements to be the two possible optimal measurements he performs in the CHSH
game when given questions $t=0$ and $t=1$ respectively. 
When considering complementarity between the measurements labelled by $t=0$ and $t=1$, we are interested in the 
probabilities $\ppost{b'| t'}$ of decoding the second bit from the post-measurement state $\tau_{t,b}$, after 
Bob performed the measurement $t$ and obtained outcome $b$.
We say that there is no-complementarity if $\ppost{b'|t'} = p(b'|t')_{\sigma_\xstr}$ for all $t'$ and $b'$. That is,
the probabilities of obtaining outcomes $b'$ when performing $t'$ are the same as if Bob had not measured $t$ at all.
Note that if the measurements that Bob (and Alice) perform in a non-local
game had no-complementarity, 
then their statistics could be described by a LHV model, since one can assign a fixed probability
distribution to the outcomes of each measurement. As a result, if
there is no complementarity, there cannot be a violation of the CHSH inequality.

We now ask what are the allowed values for $\ppost{b'|t'}$ subject to the restrictions imposed by Eq.~\ref{eq:uncertComplement} and no-signalling? 
For the case of CHSH where the notions of uncertainty and non-locality are equivalent, 
the no-signalling principle dictates that Bob can never learn the parity of the string $\xstr$
as this would allow him to determine Alice's measurement setting $s$. Since Bob might use his two measurements to determine
the parity, this imposes an additional constraint on the allowed probabilities $\ppost{b'|t'}$. For clarity, we will write
$p(right|t,\xstr) = p(b = \xstr^{(t)}|t)_{\sigma_{\xstr}}$ for the probability that Bob correctly retrieves the $t$-th bit
of the string $\xstr$ on the initial state, and 
$p(right|t',right) = p(\xstr^{(t')}|t')_{\tau_{\xstr^{(t)},t}}$ and  
$p(right|t',wrong) = p(\xstr^{(t')}|t')_{\tau_{\xstr^{(1-t)},t}}$ for the probabilities that he correctly retrieves the $t'$-th bit
given that he previously retrieved the $t$-th bit correctly or incorrectly respectively.
For any $t \neq t'$, the fact that Bob can not learn the parity by performing the two measurements in succession can then be expressed as
\begin{align}\label{eq:parityCondition}
&p(right|t,\xstr) p(right|t',right) +
&p(wrong|t,\xstr) p(wrong|t',wrong) = \frac{1}{2}\ 
\end{align}
since retrieving both bits incorrectly will also lead him to correctly 
guess the parity. The tradeoff between $p(right|t',right)$ and $p(wrong|t',wrong)$ dictated by Eq.~\ref{eq:parityCondition} is captured by Figure~\ref{fig:parityTradeoff}.
\begin{figure}[h]
\begin{center}
\includegraphics[scale=1]{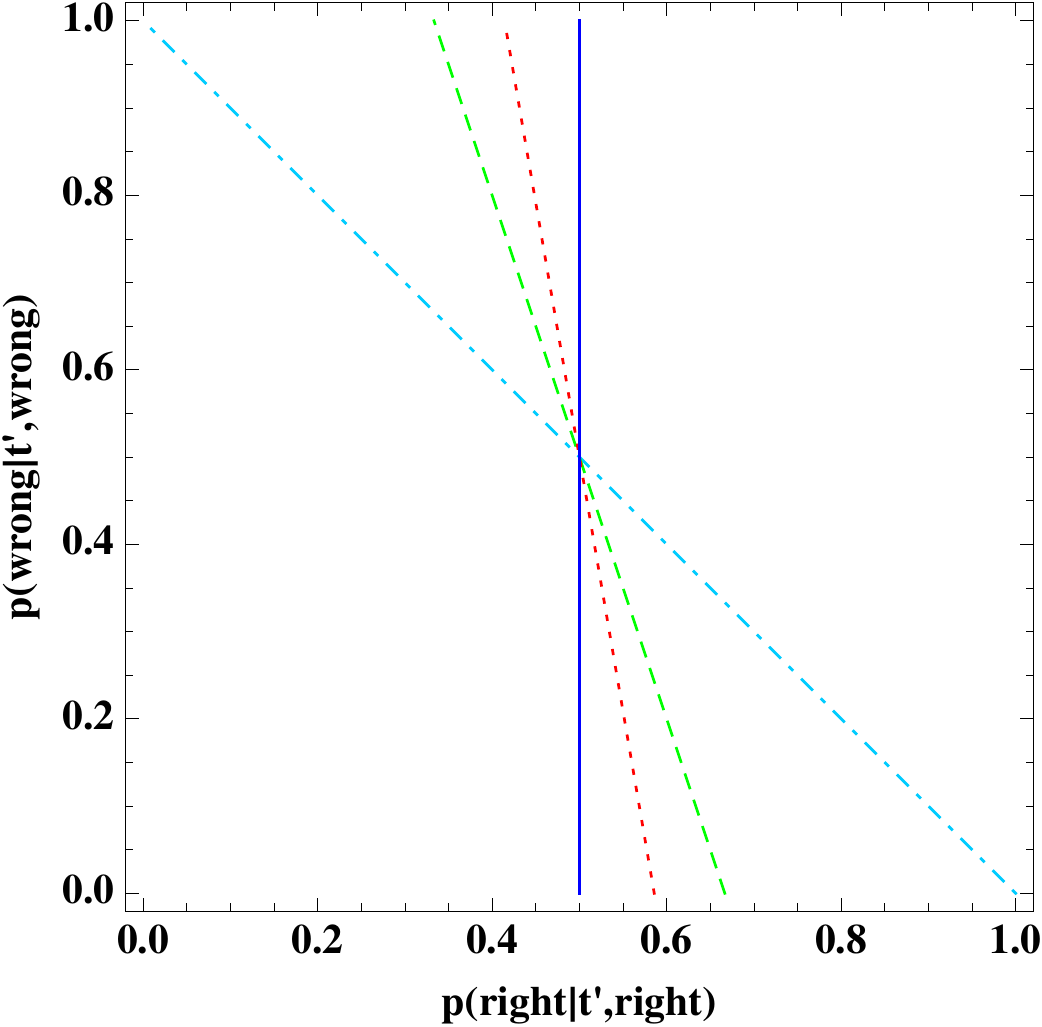}
\caption{Allowed values for $(p(right|t',right),p(wrong|t',wrong))$ for 
$p(right|t,\xstr) = 1/2$ (dot dashed light blue line), $3/4$ (dashed green line), $1/2 + 1/(2\sqrt{2})$ (dotted red line) and $1$ (solid blue line).}
\label{fig:parityTradeoff}
\end{center}
\end{figure}

As previously observed, the ``amount'' of uncertainty $\zeta_\xstr$ directly determines the violation of the CHSH inequality. In the
quantum setting we have for all $\xstr$ that $\zeta_\xstr = 1/2 + 1/(2\sqrt{2}) \approx 0.853$, where for the individual probabilities
we have for all $s$, $a$, $t$ and $b$
\begin{align}\label{eq:right}
p(right|t,\xstr) = \frac{1}{2} + \frac{1}{2 \sqrt{2}}\ .
\end{align}
We also have that $p(right|t',right) = p(right|t',wrong) = 1/2$ for all $t' \neq t$. That is Bob's optimal measurements
are maximally complementary: Once Bob performed the first measurement, he can do no better than to guess the second bit.
Note, however, that neither Eq.~\ref{eq:parityCondition} nor letting $\zeta_\xstr = 1/2+1/(2\sqrt{2})$ demands that quantum mechanics
be maximally complementary without referring to the Hilbert space formalism. In particular, consider the average probability
that Bob retrieves the $t'$-th bit correctly after the measurement $t$ has already been performed which is given by 
\begin{align}\label{eq:second}
p_{\rm second} &= p(right|t,\xstr) p(right|t',right) +
 &p(wrong|t,\xstr) p(right|t',wrong)\ .
\end{align}
We can use Eq.~\ref{eq:parityCondition} to determine
$p(right|t',wrong) = 1 - p(wrong|t',wrong)$.
Using Eq.~\ref{eq:right} we can now maximize $p_{\rm second}$ over the only free remaining variable $p(right|t',right)$ such that
$0 \leq p(wrong|t',wrong) \leq 1$.
This gives us $p_{\rm second} = 1 - 1/(2\sqrt{2}) \approx 0.65$ which is attained for $p(right|t',right) = 2 - \sqrt{2} \approx 0.59$ and
$p(right|t',wrong) = 1$. However, for the measurements used by Bob's optimal quantum strategy we only have $p_{\rm second} = 1/2$.
We thus see that it may be possible to have a physical theory which is as non-local and uncertain as quantum mechanics, but at the same
time less complementary. We would like to emphasize, however, that in quantum theory it is known~\cite{peres:book} that Bob's observables
must be maximally complementary in order to achieve Tsirelson's bound, which is a consequence of the Hilbert space formalism.

Another interesting example is the case of a PR-boxes and other less
non-local boxes. Here, we can have $p(right|t,\xstr) = 1 - \eps$ for all $t$, $\xstr$ and $1/2 \leq\eps \leq 1$. Note
that for $\eps \rightarrow 1$, maximising Eq.~\ref{eq:second} gives us
$p_{\rm second} = 1/2$ with 
$p(right|t',right) = 1/2$. 
Figure~\ref{fig:complementCurve} shows the value of $p_{\rm second}$ in terms of
$p(right|t,\xstr)$.
If there is no uncertainty, we thus have maximal complementarity. 
However, as 
we saw from the example of quantum mechanics, for intermediate values of uncertainty, the degree of complementarity is not uniquely determined.  
\begin{figure}[h]
\begin{center}
\includegraphics[scale=1]{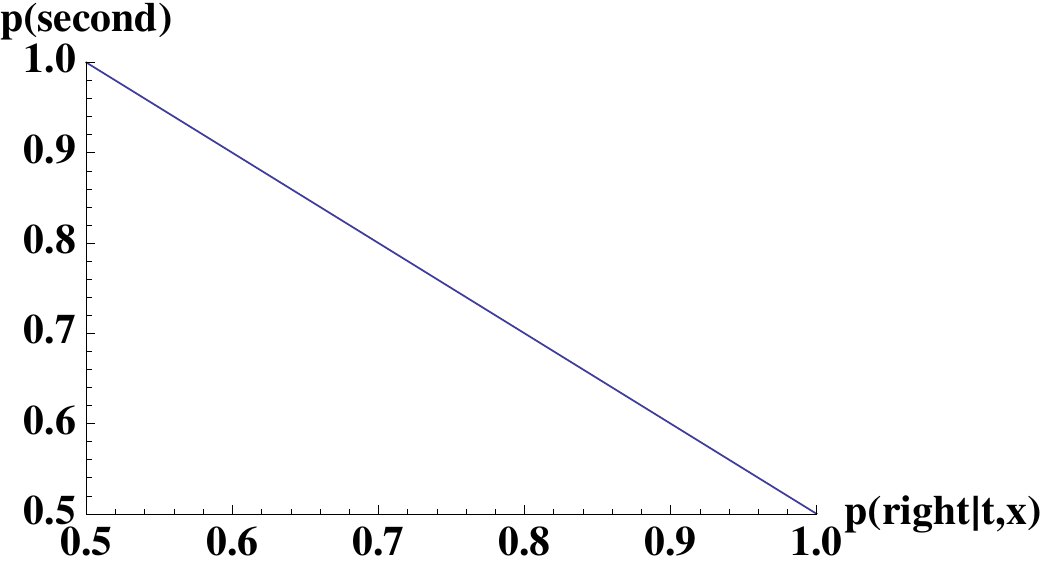}
\caption{Maximum value of $p_{\rm second}$ subject to Eq.~\ref{eq:parityCondition} in terms of $p(right|t,\xstr)$.}
\label{fig:complementCurve}
\end{center}
\end{figure}

It would be interesting to examine the role of comlementarity and its relation to non-locality in more detail. 
The existance of general monogamy relations~\cite{ben:monogamy} for example could be understood as a special form of complementarity when particular measurements are applied.

\section{From uncertainty relations to non-local games}
\label{sec:URtoNL}

In the body of the paper, we showed how every game corresponds to an uncertainty relation.  
Here we show that every uncertainty relation gives rise to a game.
To construct the game from an uncertainty relation, we can simply go the other way.

We start with the set of inequalities
\begin{align}\label{eq:urRelations-copy}
\cU = \left\{\sum_{t=1}^{n}  p(t)\ p(b_t|t)_\sigma\ \leq \zeta_\str \mid \forall \str \in \mB^{\times n} \right\}\ 
\nonumber
\end{align}
which form the uncertainty relation. Additionally, they 
can be thought of
as the average success probability
that Bob will be able to correctly output the $t$'th entry from
a string $\str$ \footnote{Indeed, the left hand side of the inequality corresponds
to a set of operators $\sum_{t=1}^{n}  p(t)\ M^{b_t}_t$ which together form a complete Postive Operator-Valued Measure (POVM) and can be used to measure this
average.}.

Here, he is performing 
a measurement on the state $\sigma$ which encodes the string $\str$.
The maximisation of these relations
\begin{align}
\psuc_\Sigma(\str)=
\zetamax
= \max_{\sigma \in \Sigma} \sum_{t=1}^n p(t) p(b_t|t)_\sigma\ ,
\end{align}
will play a key role, where here, 
the maximisation is taken over a set $\Sigma$ determined
by the theory's steering properties. 

Consider the set of all strings $\str$ induced by the above uncertainty relation.  
We construct the game by choosing some 
partitioning of these strings
into sets $P_1,\ldots,P_M$ such that
$\cup P_s = \setB^{\times L}$.
We we will challenge
Alice to output one of the strings in set $s$ and Bob to output a certain entry $t$ of that 
string. 
Alice's settings are given by $\setS = \{1,\ldots,M\}$, that is, each setting will 
correspond to a set $P_s$. The outcomes for a setting $s \in \setS$ are simply the strings contained in $P_s$, which using
the notation from the main paper we will label $\xstr$. 

Bob's settings in the game are in one-to-one correspondence to the measurements for
which we have an uncertainty relation. That is, we label his settings by his choices of measurements $\setT = \{1,\ldots,n\}$ 
and his outcomes by the alphabet of the string. We furthermore, let the distributions over measurements be given by $p(t)$
as in the case of the uncertainty relation. Note that the distribution $p(s)$ over Alice's measurement settings is not yet defined and may 
be chosen arbitrarily.
The predicate is now simply defined as $V(a,b|s,t) = 1$ 
if and only if $b_t$ is the $t$'th entry in the string $\xstr$.

Note that due to our construction, Alice will be able to steer Bob's part of the state
into some state $\sigma_{\xstr}$ encoding the string $\str$ for all settings $s$. 
We are then interested in the set of ensembles  
$\mathcal{I} = \{\mathcal{E}_s\}_s$ 
that the theory allows Alice to steer to.  In no-signalling theories,
this corresponds to
probability distributions 
$  \{p(\str |s)\}_s$ and an average state $\rho$ such that
$$
\sum_{\str \in P_s} p(\str|s)\ \sigma_\str = \rho
$$
As before, the states we can steer to in a particular ensemble $s$, 
we denote by $\Sigma_s$.

Now that we have defined the game this way, we obtain the following lower bound on the value 
$\omega(G)$ of the game.
\begin{align}
\sum_{s} \sum_{a} p(s) p({\xstr}|s) \psuc_{\Sigma_s}(\xstr) \leq \pgameMAX\ . 
\end{align}
Whereas the measurements of our uncertainty relation do form a possible strategy for Alice and Bob, which due to the steering property
can indeed be implemented, they may not be optimal. 
Indeed, there may exist an altogether different strategy consisting of different
measurements and a different state in possibly much larger dimension that enables them to do significantly better.

However, having defined the game, we may now again consider a new uncertainty relation in terms of the optimal states and measurements for this 
game. For these optimal measurements, a violation of the uncertainty relations then lead to a violation of the corresponding Tsirelson's bound and vice versa.

\end{document}